\documentclass[12pt]{article}%
\usepackage{amsmath}
\usepackage{amsfonts}
\usepackage{mitpress}
\usepackage{amssymb}
\usepackage{graphicx}%
\providecommand{\U}[1]{\protect\rule{.1in}{.1in}}
\newtheorem{theorem}{Theorem}

\newtheorem{definition}[theorem]{Definition}

\newtheorem{lemma}[theorem]{Lemma}
\newtheorem{notation}[theorem]{Notation}

\newtheorem{remark}[theorem]{Remark}

\newenvironment{proof}[1][Proof]{\noindent\textbf{#1.} }{\ \rule{0.5em}{0.5em}}
\newdimen\dummy
\dummy=\oddsidemargin
\begin{document}

\title{On the serial connection of the regular asynchronous systems}
\author{Serban E. Vlad\\Oradea City Hall, Piata Unirii Nr. 1, 410100, Oradea, Romania, serban\_e\_vlad@yahoo.com}
\maketitle

\begin{abstract}
The asynchronous systems $f$ are multi-valued functions, representing the
non-deterministic models of the asynchronous circuits from the digital
electrical engineering. In real time, they map an 'admissible input' function
$u:\mathbf{R}\rightarrow\{0,1\}^{m}$ to a set $f(u)$ of 'possible states'
$x\in f(u),$ where $x:\mathbf{R}\rightarrow\{0,1\}^{n}.$ When $f$ is defined
by making use of a 'generator function' $\Phi:\{0,1\}^{n}\times\{0,1\}^{m}%
\rightarrow\{0,1\}^{n}$, the system is called regular. The usual definition of
the serial connection of systems as composition of multi-valued functions does
not bring the regular systems into regular systems, thus the first issue in
this study is to modify in an acceptable manner the definition of the serial
connection in a way that matches regularity. This intention was expressed for
the first time, without proving the regularity of the serial connection of
systems, in the work \textit{\cite{bib1}}. Our present purpose is to restate
with certain corrections and prove Theorem 45 from that work.

\end{abstract}

\textbf{Keywords:} serial connection, asynchronous system.

\textbf{2010 MSC:} 94C99.

\section{Introduction}

The regular asynchronous systems are the Boolean dynamical systems. They
represent the (real time or discrete time) models of the asynchronous circuits
from digital electrical engineering. In \cite{bib1} we have shown that the
subsystems of the regular systems are regular, the dual systems of the regular
systems are regular, the Cartesian products of the regular systems are
regular, the parallel and the serial connections of the regular systems are
regular, the intersections and the unions of the regular systems are regular.
The result concerning the serial connections given in Theorem 45 from
\cite{bib1} was not proved by that time and our initial purpose was to give
its proof. Reconsidering the problem (in a slightly different approach) showed
that certain corrections were also necessary. The main result is represented
by Theorem \ref{The22}.

\section{Preliminaries}

\begin{notation}
We denote with $\mathbf{B}=\{0,1\}$ the binary Boole algebra, endowed with the
usual laws '$\overline{\;\;}$' complement, ' $\cdot$ '$\,$intersection,
'$\cup$'$~$union and '$\oplus$' exclusive union.
\end{notation}

\begin{definition}
Let $x:\mathbf{R}\rightarrow\mathbf{B}^{n},$ $y:\mathbf{R}\rightarrow
\mathbf{B}^{p}$ be two functions. We define the \textbf{Cartesian product}
$(x,y)$ of $x$ and $y$ by $(x,y):\mathbf{R}\rightarrow\mathbf{B}^{n+p}$,
$\forall i\in\{1,...,n+p\},\forall t\in\mathbf{R},$%
\[
(x,y)_{i}(t)=\left\{
\begin{array}
[c]{c}%
x_{i}(t),if\;i\in\{1,...,n\},\;\;\;\;\;\;\;\;\\
y_{i}(t),if\;i\in\{n+1,...,n+p\}
\end{array}
\right.  .
\]

\end{definition}

\begin{remark}
We use to identify $\mathbf{B}^{n+p}$ and $\mathbf{B}^{n}\times\mathbf{B}%
^{p}.$ This identification gives us the possibility to write: $\forall
t\in\mathbf{R},$%
\begin{equation}
(x,y)(t)=(x(t),y(t)). \label{sc24}%
\end{equation}

\end{remark}

\begin{notation}
We denote with $\chi_{H}:\mathbf{R}\rightarrow\mathbf{B}$ the characteristic
function of the set $H\subset\mathbf{R}:\forall t\in\mathbf{R},$%
\[
\chi_{H}(t)=\left\{
\begin{array}
[c]{c}%
1,if\;t\in H,\\
0,if\;t\notin H
\end{array}
\right.  .
\]

\end{notation}

\begin{notation}
$Seq$ denotes the set of the sequences $t_{0}<t_{1}<t_{2}<...$ of real numbers
that are unbounded from above. The elements of $Seq$ are usually denoted with
$(t_{k}),(t_{k}^{\prime}),...$
\end{notation}

\begin{definition}
The function $x:\mathbf{R}\rightarrow\mathbf{B}^{n}$ is called \textbf{signal}
if $\mu\in\mathbf{B}^{n}$ and $(t_{k})\in Seq$ exist such that%
\begin{equation}
x(t)=\mu\cdot\chi_{(-\infty,t_{0})}(t)\oplus x(t_{0})\cdot\chi_{\lbrack
t_{0},t_{1})}(t)\oplus...\oplus x(t_{k})\cdot\chi_{\lbrack t_{k},t_{k+1}%
)}(t)\oplus... \label{sc29}%
\end{equation}
The set of the signals is denoted by $S^{(n)}.$
\end{definition}

\begin{definition}
The function $\rho:\mathbf{R}\rightarrow\mathbf{B}^{n}$ is called
\textbf{progressive} if $(t_{k})\in Seq$ exists such that%
\begin{equation}
\rho(t)=\rho(t_{0})\cdot\chi_{\{t_{0}\}}(t)\oplus\rho(t_{1})\cdot\chi
_{\{t_{1}\}}(t)\oplus...\oplus\rho(t_{k})\cdot\chi_{\{t_{k}\}}(t)\oplus...
\label{sc30}%
\end{equation}
and $\forall i\in\{1,...,n\},$ the set $\{k|k\in\mathbf{N},\rho_{i}%
(t_{k})=1\}$ is infinite. The set of the progressive functions is denoted with
$P_{n}.$
\end{definition}

\begin{theorem}
a) If $x\in S^{(n)},y\in S^{(p)}$ then $(x,y)\in S^{(n+p)};$

b) If $\rho\in P_{n},\widetilde{\rho}\in P_{p}$ then $(\rho,\widetilde{\rho
})\in P_{n+p}.$
\end{theorem}

\begin{proof}
b) We take arbitrarily $\rho\in P_{n},\widetilde{\rho}\in P_{p}$ for which
$\forall t\in\mathbf{R},$%
\begin{equation}
\rho(t)=\rho(t_{0}^{\prime})\cdot\chi_{\{t_{0}^{\prime}\}}(t)\oplus\rho
(t_{1}^{\prime})\cdot\chi_{\{t_{1}^{\prime}\}}(t)\oplus...\oplus\rho
(t_{k}^{\prime})\cdot\chi_{\{t_{k}^{\prime}\}}(t)\oplus... \label{sc31}%
\end{equation}%
\begin{equation}
\widetilde{\rho}(t)=\widetilde{\rho}(t_{0}^{\prime\prime})\cdot\chi
_{\{t_{0}^{\prime\prime}\}}(t)\oplus\widetilde{\rho}(t_{1}^{\prime\prime
})\cdot\chi_{\{t_{1}^{\prime\prime}\}}(t)\oplus...\oplus\widetilde{\rho}%
(t_{k}^{\prime\prime})\cdot\chi_{\{t_{k}^{\prime\prime}\}}(t)\oplus...
\label{sc32}%
\end{equation}
with $(t_{k}^{\prime}),(t_{k}^{\prime\prime})\in Seq.$ We denote by
$(t_{k})\in Seq$ the sequence obtained by indexing increasingly the elements
of the set $\{t_{k}^{\prime}|k\in\mathbf{N}\}\cup\{t_{k}^{\prime\prime}%
|k\in\mathbf{N}\}.$ Equations (\ref{sc31}), (\ref{sc32}) may be rewritten
under the form%
\begin{equation}
\rho(t)=\rho(t_{0})\cdot\chi_{\{t_{0}\}}(t)\oplus\rho(t_{1})\cdot\chi
_{\{t_{1}\}}(t)\oplus...\oplus\rho(t_{k})\cdot\chi_{\{t_{k}\}}(t)\oplus...
\end{equation}%
\begin{equation}
\widetilde{\rho}(t)=\widetilde{\rho}(t_{0})\cdot\chi_{\{t_{0}\}}%
(t)\oplus\widetilde{\rho}(t_{1})\cdot\chi_{\{t_{1}\}}(t)\oplus...\oplus
\widetilde{\rho}(t_{k})\cdot\chi_{\{t_{k}\}}(t)\oplus...
\end{equation}
and we get%
\begin{equation}
(\rho,\widetilde{\rho})(t)=(\rho(t_{0}),\widetilde{\rho}(t_{0}))\cdot
\chi_{\{t_{0}\}}(t)\oplus(\rho(t_{1}),\widetilde{\rho}(t_{1}))\cdot
\chi_{\{t_{1}\}}(t)\oplus... \label{sc1_}%
\end{equation}%
\[
...\oplus(\rho(t_{k}),\widetilde{\rho}(t_{k}))\cdot\chi_{\{t_{k}\}}%
(t)\oplus...
\]
The sets%
\[
\{k|k\in\mathbf{N},(\rho,\widetilde{\rho})_{i}(t_{k})=1\}=\{k|k\in
\mathbf{N},\rho_{i}(t_{k})=1\},i=\overline{1,n},
\]%
\[
\{k|k\in\mathbf{N},(\rho,\widetilde{\rho})_{i}(t_{k})=1\}=\{k|k\in
\mathbf{N},\widetilde{\rho}_{i}(t_{k})=1\},i=\overline{n+1,n+p}%
\]
are infinite. We conclude that $(\rho,\widetilde{\rho})\in P_{n+p}.$
\end{proof}

\section{Regular systems}

\begin{notation}
$P^{\ast}(H)$ is the notation of the non-empty subsets of $H.$ In this paper
$H\in\{\mathbf{B}^{n},S^{(n)},P_{n}\}.$
\end{notation}

\begin{definition}
A function $f:U\rightarrow P^{\ast}(S^{(n)}),U\in P^{\ast}(S^{(m)})$ is called
(\textbf{asynchronous}) \textbf{system}. Any $u\in U$ is called
(\textbf{admissible}) \textbf{input} and any $x\in f(u)$ is called
(\textbf{possible}) \textbf{state}.
\end{definition}

\begin{definition}
We consider the function $\Phi:\mathbf{B}^{n}\times\mathbf{B}^{m}%
\rightarrow\mathbf{B}^{n}$. For $\nu\in\mathbf{B}^{n}$, we define $\Phi^{\nu
}:\mathbf{B}^{n}\times\mathbf{B}^{m}\rightarrow\mathbf{B}^{n}$ by $\forall
\mu\in\mathbf{B}^{n},\forall\lambda\in\mathbf{B}^{m},$%
\[
\Phi^{\nu}(\mu,\lambda)=(\overline{\nu_{1}}\cdot\mu_{1}\oplus\nu_{1}\cdot
\Phi_{1}(\mu,\lambda),...,\overline{\nu_{n}}\cdot\mu_{n}\oplus\nu_{n}\cdot
\Phi_{n}(\mu,\lambda)).
\]

\end{definition}

\begin{definition}
Let be $\mu\in\mathbf{B}^{n},U\in P^{\ast}(S^{(m)}),u\in U$ and $\rho\in
P_{n},$%
\[
u(t)=\lambda\cdot\chi_{(-\infty,t_{0}^{\prime})}(t)\oplus u(t_{0}^{\prime
})\cdot\chi_{\lbrack t_{0}^{\prime},t_{1}^{\prime})}(t)\oplus...\oplus
u(t_{k}^{\prime})\cdot\chi_{\lbrack t_{k}^{\prime},t_{k+1}^{\prime})}%
(t)\oplus...
\]%
\[
\rho(t)=\rho(t_{0}^{\prime\prime})\cdot\chi_{\{t_{0}^{\prime\prime}%
\}}(t)\oplus\rho(t_{1}^{\prime\prime})\cdot\chi_{\{t_{1}^{\prime\prime}%
\}}(t)\oplus...\oplus\rho(t_{k}^{\prime\prime})\cdot\chi_{\{t_{k}%
^{\prime\prime}\}}(t)\oplus...
\]
with $\lambda\in\mathbf{B}^{m}$ and $(t_{k}^{\prime}),(t_{k}^{\prime\prime
})\in Seq.$ The \textbf{orbit} $\Phi^{\rho}(\mu,u,\cdot)\in S^{(n)}$ is
defined like this. We denote by $(t_{k})\in Seq$ the elements of the set
$\{t_{k}^{\prime}|k\in\mathbf{N}\}\cup\{t_{k}^{\prime\prime}|k\in\mathbf{N}\}$
indexed increasingly. Then $\forall t\in\mathbf{R},$%
\[
\Phi^{\rho}(\mu,u,t)=\mu\cdot\chi_{(-\infty,t_{0})}(t)\oplus\omega_{0}%
\cdot\chi_{\lbrack t_{0},t_{1})}(t)\oplus...\oplus\omega_{k}\cdot\chi_{\lbrack
t_{k},t_{k+1})}(t)\oplus...
\]
where%
\[
\omega_{0}=\Phi^{\rho(t_{0})}(\mu,u(t_{0})),
\]%
\[
\omega_{k+1}=\Phi^{\rho(t_{k+1})}(\omega_{k},u(t_{k+1})),k\in\mathbf{N.}%
\]

\end{definition}

\begin{definition}
$f$ is called \textbf{regular asynchronous system} if $\Phi$ and the functions
$i_{f}:U\rightarrow P^{\ast}(\mathbf{B}^{n}),$ $\pi_{f}:\Delta_{f}\rightarrow
P^{\ast}(P_{n})$ exist,%
\[
\Delta_{f}=\{(\mu,u)|u\in U,\mu\in i_{f}(u)\}
\]
such that $\forall u\in U,$%
\[
f(u)=\{\Phi^{\rho}(\mu,u,\cdot)|\mu\in i_{f}(u),\rho\in\pi_{f}(\mu,u)\}.
\]
The functions $i_{f},\pi_{f}$ are called the \textbf{initial state function},
respectively the \textbf{computation function} of $f$. $\Phi$ is called the
\textbf{generator function} of $f$; we also say that $f$ \textbf{is generated}
by $\Phi.$
\end{definition}

\section{The serial connection of the regular systems}

\begin{remark}
Let be the systems $f:U\rightarrow P^{\ast}(S^{(n)}),U\in P^{\ast}(S^{(m)}) $
and $h:X\rightarrow P^{\ast}(S^{(p)}),X\in P^{\ast}(S^{(n)})$ such that
$\forall u\in U,f(u)\subset X.$ In general the serial connection $h\circ f$ of
$h$ and $f$ is defined by $h\circ f:U\rightarrow P^{\ast}(S^{(p)}),\forall
u\in U,(h\circ f)(u)=\underset{x\in f(u)}{%
{\displaystyle\bigcup}
}h(x)=\{y|x\in f(u),y\in h(x)\}$ (the composition of the multi-valued
functions). This definition does not match regularity, thus we are forced to
adopt the following definition of the serial connection, that is still acceptable.
\end{remark}

\begin{definition}
We define the \textbf{serial connection} $h\ast f$ of $h$ and $f$ by $h\ast
f:U\rightarrow P^{\ast}(S^{(n+p)}),$ $\forall u\in U,$%
\begin{equation}
(h\ast f)(u)=\{(x,y)|x\in f(u),y\in h(x)\}. \label{sc6}%
\end{equation}

\end{definition}

\begin{definition}
For the functions $\Phi:\mathbf{B}^{n}\times\mathbf{B}^{m}\rightarrow
\mathbf{B}^{n}$ and $\Psi:\mathbf{B}^{p}\times\mathbf{B}^{n}\rightarrow
\mathbf{B}^{p}$ we define $\Psi\ast\Phi:\mathbf{B}^{n+p}\times\mathbf{B}%
^{m}\rightarrow\mathbf{B}^{n+p}$ by $\forall((\mu,\delta),\lambda
)\in(\mathbf{B}^{n}\times\mathbf{B}^{p})\times\mathbf{B}^{m},$%
\begin{equation}
(\Psi\ast\Phi)((\mu,\delta),\lambda)=(\Phi(\mu,\lambda),\Psi(\delta,\Phi
(\mu,\lambda))). \label{sc25}%
\end{equation}
In the previous equation we have identified $\mathbf{B}^{n+p}$ with
$\mathbf{B}^{n}\times\mathbf{B}^{p}.$
\end{definition}

\begin{definition}
For any $\nu\in\mathbf{B}^{n},\widetilde{\nu}\in\mathbf{B}^{p}$ and $\Phi
,\Psi$ like previously, we define $(\Psi\ast\Phi)^{(\nu,\widetilde{\nu}%
)}:\mathbf{B}^{n+p}\times\mathbf{B}^{m}\rightarrow\mathbf{B}^{n+p}$ in the
following manner:%
\begin{equation}
(\Psi\ast\Phi)^{(\nu,\widetilde{\nu})}=\Psi^{\widetilde{\nu}}\ast\Phi^{\nu}.
\label{sc26}%
\end{equation}

\end{definition}

\begin{lemma}
\label{The6}We presume that $\Phi:\mathbf{B}^{n}\times\mathbf{B}%
^{m}\rightarrow\mathbf{B}^{n},$ $\Psi:\mathbf{B}^{p}\times\mathbf{B}%
^{n}\rightarrow\mathbf{B}^{p}$ as well as $\mu\in\mathbf{B}^{n},\delta
\in\mathbf{B}^{p},u\in U,\rho\in P_{n},\widetilde{\rho}\in P_{p}$ are given.
The following formula is true: $\forall t\in\mathbf{R},$%
\[
(\Phi^{\rho}(\mu,u,t),\Psi^{\widetilde{\rho}}(\delta,\Phi^{\rho}(\mu
,u,\cdot),t))=(\Psi\ast\Phi)^{(\rho,\widetilde{\rho})}((\mu,\delta),u,t).
\]

\end{lemma}

\begin{proof}
Let be $u\in U,$%
\begin{equation}
u(t)=u(-\infty+0)\cdot\chi_{(-\infty,t_{0}^{\prime})}(t)\oplus u(t_{0}%
^{\prime})\cdot\chi_{\lbrack t_{0}^{\prime},t_{1}^{\prime})}(t)\oplus...
\label{sc10}%
\end{equation}%
\[
...\oplus u(t_{k}^{\prime})\cdot\chi_{\lbrack t_{k}^{\prime},t_{k+1}^{\prime
})}(t)\oplus...
\]
together with the functions $\rho\in P_{n},\widetilde{\rho}\in P_{p},$%
\begin{equation}
\rho(t)=\rho(t_{0}^{\prime\prime})\cdot\chi_{\{t_{0}^{\prime\prime}%
\}}(t)\oplus\rho(t_{1}^{\prime\prime})\cdot\chi_{\{t_{1}^{\prime\prime}%
\}}(t)\oplus...\oplus\rho(t_{k}^{\prime\prime})\cdot\chi_{\{t_{k}%
^{\prime\prime}\}}(t)\oplus... \label{sc11}%
\end{equation}%
\begin{equation}
\widetilde{\rho}(t)=\widetilde{\rho}(t_{0}^{^{\prime\prime\prime}})\cdot
\chi_{\{t_{0}^{^{\prime\prime\prime}}\}}(t)\oplus\widetilde{\rho}%
(t_{1}^{^{\prime\prime\prime}})\cdot\chi_{\{t_{1}^{^{\prime\prime\prime}}%
\}}(t)\oplus...\oplus\widetilde{\rho}(t_{k}^{^{\prime\prime\prime}})\cdot
\chi_{\{t_{k}^{^{\prime\prime\prime}}\}}(t)\oplus... \label{sc12}%
\end{equation}
where $(t_{k}^{\prime}),(t_{k}^{\prime\prime}),(t_{k}^{\prime\prime\prime})\in
Seq.$ We denote with $(t_{k})\in Seq$ the sequence that is obtained by
indexing increasingly the elements of the set $\{t_{k}^{\prime}|k\in
\mathbf{N}\}\cup\{t_{k}^{^{\prime\prime}}|k\in\mathbf{N}\}\cup\{t_{k}%
^{^{\prime\prime\prime}}|k\in\mathbf{N}\},$ for which (\ref{sc10}),
(\ref{sc11}), (\ref{sc12}) may be rewriten under the form%
\begin{equation}
u(t)=u(-\infty+0)\cdot\chi_{(-\infty,t_{0})}(t)\oplus u(t_{0})\cdot
\chi_{\lbrack t_{0},t_{1})}(t)\oplus...
\end{equation}%
\[
...\oplus u(t_{k})\cdot\chi_{\lbrack t_{k},t_{k+1})}(t)\oplus...
\]%
\begin{equation}
\rho(t)=\rho(t_{0})\cdot\chi_{\{t_{0}\}}(t)\oplus\rho(t_{1})\cdot\chi
_{\{t_{1}\}}(t)\oplus...\oplus\rho(t_{k})\cdot\chi_{\{t_{k}\}}(t)\oplus...
\label{sc15}%
\end{equation}%
\begin{equation}
\widetilde{\rho}(t)=\widetilde{\rho}(t_{0})\cdot\chi_{\{t_{0}\}}%
(t)\oplus\widetilde{\rho}(t_{1})\cdot\chi_{\{t_{1}\}}(t)\oplus...\oplus
\widetilde{\rho}(t_{k})\cdot\chi_{\{t_{k}\}}(t)\oplus... \label{sc16}%
\end{equation}
We have%
\begin{equation}
\Phi^{\rho}(\mu,u,t)=\mu\cdot\chi_{(-\infty,t_{0})}(t)\oplus\omega_{0}%
\cdot\chi_{\lbrack t_{0},t_{1})}(t)\oplus...\oplus\omega_{k}\cdot\chi_{\lbrack
t_{k},t_{k+1})}(t)\oplus... \label{sc20}%
\end{equation}
where%
\begin{equation}
\omega_{0}=\Phi^{\rho(t_{0})}(\mu,u(t_{0})), \label{sc22}%
\end{equation}%
\begin{equation}
\omega_{k+1}=\Phi^{\rho(t_{k+1})}(\omega_{k},u(t_{k+1})),k\in\mathbf{N}
\label{sc13}%
\end{equation}
and furthermore%
\begin{equation}
\Psi^{\widetilde{\rho}}(\delta,\Phi^{\rho}(\mu,u,\cdot),t)=\delta\cdot
\chi_{(-\infty,t_{0})}(t)\oplus\gamma_{0}\cdot\chi_{\lbrack t_{0},t_{1}%
)}(t)\oplus... \label{sc21}%
\end{equation}%
\[
...\oplus\gamma_{k}\cdot\chi_{\lbrack t_{k},t_{k+1})}(t)\oplus...
\]
where%
\begin{equation}
\gamma_{0}=\Psi^{\widetilde{\rho}(t_{0})}(\delta,\omega_{0}), \label{sc23}%
\end{equation}%
\begin{equation}
\gamma_{k+1}=\Psi^{\widetilde{\rho}(t_{k+1})}(\gamma_{k},\omega_{k+1}%
),k\in\mathbf{N.} \label{sc14}%
\end{equation}
We conclude from (\ref{sc20}),(\ref{sc21}) that%
\begin{equation}
(\Phi^{\rho}(\mu,u,t),\Psi^{\widetilde{\rho}}(\delta,\Phi^{\rho}(\mu
,u,\cdot),t))=(\mu,\delta)\cdot\chi_{(-\infty,t_{0})}(t)\oplus\label{sc18}%
\end{equation}%
\[
\oplus(\omega_{0},\gamma_{0})\cdot\chi_{\lbrack t_{0},t_{1})}(t)\oplus
...\oplus(\omega_{k},\gamma_{k})\cdot\chi_{\lbrack t_{k},t_{k+1})}%
(t)\oplus...
\]
On the other hand%
\begin{equation}
(\rho,\widetilde{\rho})(t)=(\rho(t_{0}),\widetilde{\rho}(t_{0}))\cdot
\chi_{\{t_{0}\}}(t)\oplus(\rho(t_{1}),\widetilde{\rho}(t_{1}))\cdot
\chi_{\{t_{1}\}}(t)\oplus...
\end{equation}%
\[
...\oplus(\rho(t_{k}),\widetilde{\rho}(t_{k}))\cdot\chi_{\{t_{k}\}}%
(t)\oplus...
\]
thus%
\begin{equation}
(\Psi\ast\Phi)^{(\rho,\widetilde{\rho})}((\mu,\delta),u,t)=(\mu,\delta
)\cdot\chi_{(-\infty,t_{0})}(t)\oplus\theta_{0}\cdot\chi_{\lbrack t_{0}%
,t_{1})}(t)\oplus... \label{sc19}%
\end{equation}%
\[
...\oplus\theta_{k}\cdot\chi_{\lbrack t_{k},t_{k+1})}(t)\oplus...
\]
is true with%
\begin{equation}
\theta_{0}=(\Psi\ast\Phi)^{(\rho,\widetilde{\rho})(t_{0})}((\mu,\delta
),u(t_{0})), \label{sc27}%
\end{equation}%
\begin{equation}
\theta_{k+1}=(\Psi\ast\Phi)^{(\rho,\widetilde{\rho})(t_{k+1})}(\theta
_{k},u(t_{k+1})),k\in\mathbf{N.} \label{sc28}%
\end{equation}
We prove by induction on $k$ that
\begin{equation}
\theta_{k}=(\omega_{k},\gamma_{k}),k\in\mathbf{N}. \label{sc17}%
\end{equation}
We have%
\[
\theta_{0}\overset{(\ref{sc27})}{=}(\Psi\ast\Phi)^{(\rho,\widetilde{\rho
})(t_{0})}((\mu,\delta),u(t_{0}))\overset{(\ref{sc24})}{=}(\Psi\ast
\Phi)^{(\rho(t_{0}),\widetilde{\rho}(t_{0}))}((\mu,\delta),u(t_{0}))
\]%
\[
\overset{(\ref{sc26})}{=}(\Psi^{\widetilde{\rho}(t_{0})}\ast\Phi^{\rho(t_{0}%
)})((\mu,\delta),u(t_{0}))
\]%
\[
\overset{(\ref{sc25})}{=}(\Phi^{\rho(t_{0})}(\mu,u(t_{0})),\Psi^{\widetilde
{\rho}(t_{0})}(\delta,\Phi^{\rho(t_{0})}(\mu,u(t_{0}))))
\]%
\[
\overset{(\ref{sc22})}{=}(\omega_{0},\Psi^{\widetilde{\rho}(t_{0})}%
(\delta,\omega_{0}))\overset{(\ref{sc23})}{=}(\omega_{0},\gamma_{0})
\]
thus (\ref{sc17}) is true for $k=0.$ We presume that it is true for $k$ and we
prove it for $k+1:$%
\[
\theta_{k+1}\overset{(\ref{sc28})}{=}(\Psi\ast\Phi)^{(\rho,\widetilde{\rho
})(t_{k+1})}(\theta_{k},u(t_{k+1}))\overset{(\ref{sc24})}{=}(\Psi\ast
\Phi)^{(\rho(t_{k+1}),\widetilde{\rho}(t_{k+1}))}(\theta_{k},u(t_{k+1}))
\]%
\[
\overset{(\ref{sc26})}{=}(\Psi^{\widetilde{\rho}(t_{k+1})}\ast\Phi
^{\rho(t_{k+1})})(\theta_{k},u(t_{k+1}))\overset{hyp}{=}(\Psi^{\widetilde
{\rho}(t_{k+1})}\ast\Phi^{\rho(t_{k+1})})((\omega_{k},\gamma_{k}),u(t_{k+1}))
\]%
\[
\overset{(\ref{sc25})}{=}(\Phi^{\rho(t_{k+1})}(\omega_{k},u(t_{k+1}%
)),\Psi^{\widetilde{\rho}(t_{k+1})}(\gamma_{k},\Phi^{\rho(t_{k+1})}(\omega
_{k},u(t_{k+1}))))
\]%
\[
\overset{(\ref{sc13})}{=}(\omega_{k+1},\Psi^{\widetilde{\rho}(t_{k+1})}%
(\gamma_{k},\omega_{k+1}))\overset{(\ref{sc14})}{=}(\omega_{k+1},\gamma
_{k+1}).
\]
Equation (\ref{sc17}) is proved and its truth shows, from (\ref{sc18}) and
(\ref{sc19}), the validity of the statement of the Lemma.
\end{proof}

\begin{definition}
The functions $i_{f}:U\rightarrow P^{\ast}(\mathbf{B}^{n})$ and $\pi
_{f}:\Delta_{f}\rightarrow P^{\ast}(P_{n})$ are given,%
\[
\Delta_{f}=\{(\mu,u)|u\in U,\mu\in i_{f}(u)\}
\]
such that $\forall u\in U,$%
\begin{equation}
f(u)=\{\Phi^{\rho}(\mu,u,\cdot)|\mu\in i_{f}(u),\rho\in\pi_{f}(\mu,u)\}
\label{sc7}%
\end{equation}
and similarly the functions $i_{h}:X\rightarrow P^{\ast}(\mathbf{B}^{p})$ and
$\pi_{h}:\Delta_{h}\rightarrow P^{\ast}(P_{p})$ are given,%
\[
\Delta_{h}=\{(\delta,x)|x\in X,\delta\in i_{h}(x)\}
\]
such that $\forall x\in X,$%
\begin{equation}
h(x)=\{\Psi^{\widetilde{\rho}}(\delta,x,\cdot)|\delta\in i_{h}(x),\widetilde
{\rho}\in\pi_{h}(\delta,x)\}. \label{sc8}%
\end{equation}
We presume that $\forall u\in U,f(u)\subset X.$ We define $i:U\rightarrow
P^{\ast}(\mathbf{B}^{n+p})$ by $\forall u\in U,$%
\begin{equation}
i(u)=\{(\mu,\delta)|\mu\in i_{f}(u),\exists\rho^{\prime}\in\pi_{f}%
(\mu,u),\delta\in i_{h}(\Phi^{\rho^{\prime}}(\mu,u,\cdot))\} \label{sc4}%
\end{equation}
and $\pi:\Delta\rightarrow P^{\ast}(P_{n+p})$ respectively by%
\[
\Delta=\{((\mu,\delta),u)|u\in U,\mu\in i_{f}(u),\exists\rho\in\pi_{f}%
(\mu,u),\delta\in i_{h}(\Phi^{\rho}(\mu,u,\cdot))\},
\]
$\forall((\mu,\delta),u)\in\Delta,$%
\begin{equation}
\pi((\mu,\delta),u)= \label{sc5}%
\end{equation}%
\[
=\{(\rho,\widetilde{\rho})|\rho\in\pi_{f}(\mu,u),\delta\in i_{h}(\Phi^{\rho
}(\mu,u,\cdot)),\widetilde{\rho}\in\pi_{h}(\delta,\Phi^{\rho}(\mu
,u,\cdot))\}.
\]

\end{definition}

\begin{lemma}
\label{The8}The following equality holds for $u\in U:$%
\[
\{(\Psi\ast\Phi)^{(\rho,\widetilde{\rho})}((\mu,\delta),u,\cdot)|\mu\in
i_{f}(u),\rho\in\pi_{f}(\mu,u),
\]%
\[
\delta\in i_{h}(\Phi^{\rho}(\mu,u,\cdot)),\widetilde{\rho}\in\pi_{h}%
(\delta,\Phi^{\rho}(\mu,u,\cdot))\}
\]%
\[
=\{(\Psi\ast\Phi)^{(\rho,\widetilde{\rho})}((\mu,\delta),u,\cdot)|(\mu
,\delta)\in i(u),(\rho,\widetilde{\rho})\in\pi((\mu,\delta),u)\}.
\]

\end{lemma}

\begin{proof}
We denote for $u\in U$%
\[
A=\{(\Psi\ast\Phi)^{(\rho,\widetilde{\rho})}((\mu,\delta),u,\cdot)|\mu\in
i_{f}(u),\rho\in\pi_{f}(\mu,u),
\]%
\[
\delta\in i_{h}(\Phi^{\rho}(\mu,u,\cdot)),\widetilde{\rho}\in\pi_{h}%
(\delta,\Phi^{\rho}(\mu,u,\cdot))\},
\]%
\[
B=\{(\Psi\ast\Phi)^{(\rho,\widetilde{\rho})}((\mu,\delta),u,\cdot)|(\mu
,\delta)\in i(u),(\rho,\widetilde{\rho})\in\pi((\mu,\delta),u)\}.
\]
Let $(\Psi\ast\Phi)^{(\rho,\widetilde{\rho})}((\mu,\delta),u,\cdot)\in A$ be
arbitrary, where%
\begin{equation}
\mu\in i_{f}(u),\rho\in\pi_{f}(\mu,u),\delta\in i_{h}(\Phi^{\rho}(\mu
,u,\cdot)),\widetilde{\rho}\in\pi_{h}(\delta,\Phi^{\rho}(\mu,u,\cdot)).
\label{sc2}%
\end{equation}
We infer that%
\begin{equation}
\mu\in i_{f}(u),\exists\rho^{\prime}\in\pi_{f}(\mu,u),\delta\in i_{h}%
(\Phi^{\rho^{\prime}}(\mu,u,\cdot)), \label{sc3}%
\end{equation}%
\[
\rho\in\pi_{f}(\mu,u),\delta\in i_{h}(\Phi^{\rho}(\mu,u,\cdot)),\widetilde
{\rho}\in\pi_{h}(\delta,\Phi^{\rho}(\mu,u,\cdot))
\]
holds, from (\ref{sc4}), (\ref{sc5}) we get%
\begin{equation}
(\mu,\delta)\in i(u),(\rho,\widetilde{\rho})\in\pi((\mu,\delta),u),
\label{sc9}%
\end{equation}
thus $(\Psi\ast\Phi)^{(\rho,\widetilde{\rho})}((\mu,\delta),u,\cdot)\in B$ and
finally $A\subset B.$

Conversely, let $(\Psi\ast\Phi)^{(\rho,\widetilde{\rho})}((\mu,\delta
),u,\cdot)\in B$ be arbitrary, with (\ref{sc9}) fulfilled, wherefrom we get
that (\ref{sc3}) is true. Then (\ref{sc2}) holds, meaning that $B\subset A.$

The statement of the Theorem is proved.
\end{proof}

\begin{theorem}
\label{The22}We have that $\forall u\in U,$%
\[
(h\ast f)(u)=\{(\Psi\ast\Phi)^{(\rho,\widetilde{\rho})}((\mu,\delta
),u,\cdot)|(\mu,\delta)\in i(u),(\rho,\widetilde{\rho})\in\pi((\mu
,\delta),u)\},
\]
i.e. $h\ast f$ is regular generated by $\Psi\ast\Phi,$ $i=i_{h\ast f}$ and
$\pi=\pi_{h\ast f}.$
\end{theorem}

\begin{proof}
Let $u\in U$ be arbitrary. We infer:%
\[
(h\ast f)(u)\overset{(\ref{sc6})}{=}\{(x,y)|x\in f(u),y\in h(x)\}
\]%
\[
\overset{(\ref{sc7}),(\ref{sc8})}{=}\{(\Phi^{\rho}(\mu,u,\cdot),\Psi
^{\widetilde{\rho}}(\delta,\Phi^{\rho}(\mu,u,\cdot),\cdot)|\mu\in
i_{f}(u),\rho\in\pi_{f}(\mu,u),
\]%
\[
\delta\in i_{h}(\Phi^{\rho}(\mu,u,\cdot)),\widetilde{\rho}\in\pi_{h}%
(\delta,\Phi^{\rho}(\mu,u,\cdot))\}
\]%
\[
\overset{Lemma\;\ref{The6}}{=}\{(\Psi\ast\Phi)^{(\rho,\widetilde{\rho})}%
((\mu,\delta),u,\cdot)|\mu\in i_{f}(u),\rho\in\pi_{f}(\mu,u),
\]%
\[
\delta\in i_{h}(\Phi^{\rho}(\mu,u,\cdot)),\widetilde{\rho}\in\pi_{h}%
(\delta,\Phi^{\rho}(\mu,u,\cdot))\}
\]%
\[
\overset{Lemma\;\ref{The8}}{=}\{(\Psi\ast\Phi)^{(\rho,\widetilde{\rho})}%
((\mu,\delta),u,\cdot)|(\mu,\delta)\in i(u),(\rho,\widetilde{\rho})\in\pi
((\mu,\delta),u)\}.
\]

\end{proof}

\end{document}